\documentclass[global,final]{svjour}

\journalname{Algorithmica}
\usepackage{amscd,amssymb}
\usepackage{mathrsfs} 
\usepackage{graphicx,url}

\newcommand{\abs}[1]{\left\vert#1\right\vert}

\def\I {{\mathbb I}}

\def\N{{\mathbb N}}

\def\R{{\mathbb R}}
\def\s{{\mathbb S}}
\def\e{\varepsilon}
\def\ve{\varepsilon}
\def\VC{{\mathrm{VC}}}

\newtheorem{algorithm}[theorem]{Algorithm}

\begin{document}

\pagestyle{headings}

\title{
Lower Bounds on Performance of Metric Tree Indexing Schemes for Exact Similarity Search in High Dimensions}

\titlerunning{Lower Bounds for Metric Trees}

\author{Vladimir Pestov}

\authorrunning{V. Pestov}
\institute{Department of Mathematics and Statistics, University of Ottawa, 
\newline
585 King Edward Avenue, Ottawa, Ontario K1N 6N5 Canada 
\newline
\email{vpest283@uottawa.ca}
}
%


\maketitle

\begin{abstract}
  \noindent 
Within a mathematically rigorous model, we analyse the curse of dimensionality for deterministic exact similarity search in the context of popular indexing schemes: metric trees. The datasets $X$ are sampled randomly from a domain $\Omega$, equipped with a distance, $\rho$, and an underlying probability distribution, $\mu$.
While performing an asymptotic analysis, we send the intrinsic dimension $d$ of $\Omega$ to infinity, and assume that the size of a dataset, $n$, grows superpolynomially yet subexponentially in $d$. 
Exact similarity search refers to finding the nearest neighbour in the dataset $X$ to a query point $\omega\in\Omega$, where the query points are subject to the same probability distribution $\mu$ as datapoints. 
Let $\mathscr F$ denote a class of all $1$-Lipschitz functions on $\Omega$ that can be used as decision functions in constructing a hierarchical metric tree indexing scheme. Suppose the VC dimension of the class of all sets $\{\omega\colon f(\omega)\geq a\}$, $a\in\R$
is $o(n^{1/4}/\log^2n)$.
(In view of a 1995 result of Goldberg and Jerrum, even a stronger complexity assumption $d^{O(1)}$ is reasonable.)
We deduce the $\Omega(n^{1/4})$ lower bound on the expected average case performance of hierarchical metric-tree based indexing schemes for exact similarity search in $(\Omega,X)$. In paricular, this bound is superpolynomial in $d$.
\end{abstract}



\section*{Introduction}\label{S:one}
Every similarity query in a dataset with $n$ points can be answered in time $O(n)$ through a simple linear scan, and in practice such a scan sometimes outperforms the best known indexing schemes for high-dimensional workloads. This is known as the {\em curse of dimensionality,} cf. e.g. Chapter 9 in \cite{santini}, as well as \cite{BGRS,WSB}. 

Paradoxically, there is no known mathematical proof that the above phenomenon is in the nature of high-dimensional datasets. While the concept of intrinsic dimension of data is open to a discussion (see e.g. \cite{clarkson:06,pestov:08}), even in cases commonly accepted as ``high-dimensional'' (e.g. uniformly distributed data in the Hamming cube $\{0,1\}^d$ as $d\to\infty$), the ``curse of dimensionality conjecture'' for proximity search remains unproven \cite{indyk:04}.  Diverse results in this direction  \cite{borodin:99,barkol:00,chavez:01,shaft:06,AIP,PT,PTW,VolPest09} are still preliminary.

Here we will verify the curse of dimensionality for a particular class of indexing schemes widely used in similarity search and going back to \cite{uhlmann:91}: metric trees. So are called hierarchical partitioning indexing schemes equipped with 1-Lipschitz (non-expanding) decision functions $f_C$ at every inner node $C$. The value of $f_C$ at the query point $q$ determines which child node to follow. If $f_C(q)>\ve$, where $\ve>0$ is the range query radius, we can be sure that the solution to the range similarity problem is not in the region $C_-=\{x\colon f_C(x)\leq 0\}$. Similarly, for $f_C(q)<-\ve$. However, if $q$ lies in the decision margin $\{-\ve\leq f_C\leq \ve\}$, no child node can be discarded, and branching occurs.

Choosing a decision function when an indexing scheme is being constructed thus becomes an unsupervised soft margin classification problem. 

\begin{figure}[ht]
\centering
\includegraphics[width=.5\textwidth]
{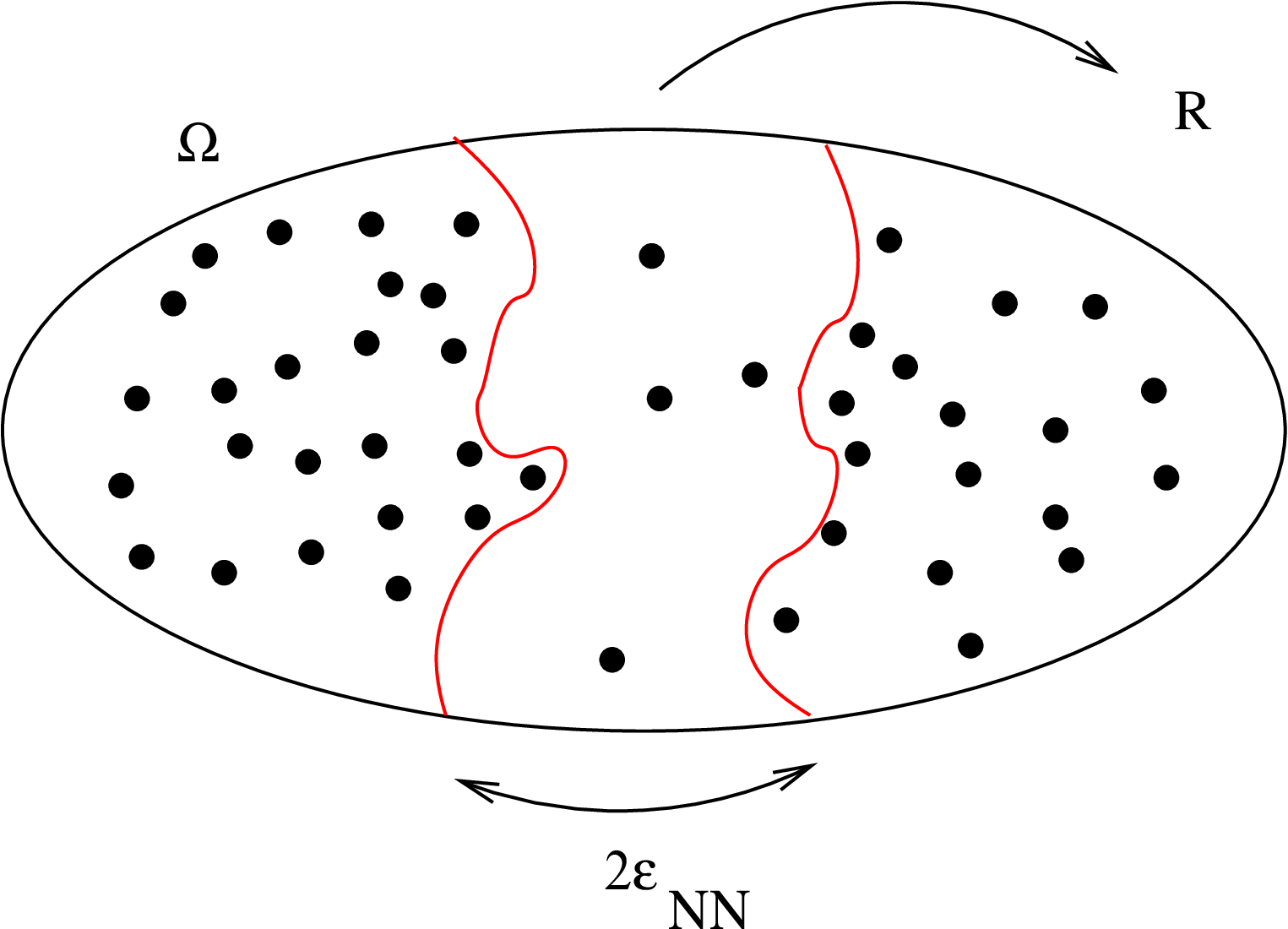}
\caption{Constructing a decision function.}
\label{fig:classification2}
\end{figure}

Assuming the domain is high-dimensional, the well-known concentration of measure phenomenon implies that the measure of the margin approaches one as dimension grows. And under assumption that the combinatorial dimension of the class of all available classifiers (decision functions) grows not too fast (say, polynomially in the dimension of the domain), standard methods of statistical learning imply that randomly sampled data is concentrated on the margin as well, making efficient indexing impossible. 

To be more precise, we assume that the domain $(\Omega,\rho)$ is a metric space equipped with a probability distribution $\mu$, and that the datapoints are drawn randomly with regard to $\mu$. The intrinsic dimension of $\Omega$ is defined in terms of concentration of measure as in \cite{pestov:08}. This concept agrees with the usual notion of dimension for such important domains as the Euclidean space $\R^d$ with the gaussian measure $\gamma^d$, the cube $[0,1]^d$ with the uniform measure, the Euclidean sphere $\s^n$ with the Haar (Lebesgue) measure, and the Hamming cube $\{0,1\}^n$ with the Hamming distance and the counting measure. 
Following \cite{indyk:04}, we require the number of datapoints $n$ to grow with regard to dimension $d$ superpolynomially, yet subexponentially: $n=d^{\omega(1)}$ and $d=\omega(\log n)$.

It is clear that the computational complexity of decision functions used in constructing a metric tree is a major factor in a scheme performance. We take this into account in the form of a combinatorial restriction on the subclass $\mathscr F$ of all functions on $\Omega$ that are allowed to be used as decision functions. Namely, we require a well-known parameter of statistical learning theory, the Vapnik-Chervonenkis dimension \cite{vapnik:98}, of all binary functions of the form $\theta(f-a)$, $f\in\mathscr F$, where $\theta$ is the Heaviside function, to be $o(n^{1/4}/\log^2n)$. This is in paricular satisfied if the VC dimension in question is polynomial in $d$.
A very general class of functions satisfying this VC dimension bound is provided by a theorem of Goldberg and Jerrum \cite{GJ}, and apparently decision functions of all indexing schemes used in practice so far in Euclidean (and Hamming cube) domains fall into this class.

Under above assumptions, we prove a lower bound $\Omega(n^{1/4})$ on the expected average performance of a metric tree. This bound is in particular superpolynomial in $d$. 

It is probably hard to argue that the real data can be simulated by random sampling from a high-dimensional distribution. The present author happily concedes that implications of the above result for high-dimensional similarity search are only indirect: our work underscores the importance of further developing a relevant theory of intrinsic dimensionality of data \cite{clarkson:06}, which would equate indexability with low dimension.

A shorter conference version of the paper (with a weaker bound $d^{\omega(1)}$) appears in: Proc. 4th Int. Conf. on Similarity Search and Applications (SISAP 2011), Lipari, Italy, ACM, New York, NY, pp. 25--32. The author is thankful to the anonymous referee for a number of useful remarks, in particular the present lower bound $\Omega(n^{1/4})$ is obtained in response to one of them.

\section{General framework for similarity search}

We follow a formalism of \cite{HKMPS02} as adapted for similarity search in \cite{pestov:00,PeSt06}. A {\em workload} is a triple $W=(\Omega,X,{\mathcal Q})$, where 
$\Omega$ is the {\em domain,} whose elements can occur both as datapoints and as query points, $X\subseteq\Omega$ is a finite subset ({\em dataset}, or {\em instance}), and $\mathcal Q\subseteq 2^{\Omega}$ is a family of {\em queries.} 
{\em Answering a query} $Q\in{\mathcal Q}$  means listing all datapoints $x\in X\cap Q$. 

A ({\em dis}){\em similarity measure} on $\Omega$ is a function of two arguments $\rho\colon\Omega\times\Omega\to\R$, which we assume to be a metric, as in \cite{zezula:06}. (Sometimes one needs to consider more general similarity measures, cf. \cite{farago:93,PeSt06}.) A {\em range similarity
query centred at} $\omega\in\Omega$ is a ball of radius $\ve$ around the query point:
\[Q ={\mathcal B}_\ve(\omega)= \{x\in\Omega\colon \rho(\omega,x)<\ve\}.\]
Equipped with such balls as queries, the triple $W=(\Omega,\rho,X)$ forms a {\em range similarity workload}.

The {\em $k$-nearest neighbours} ($k$-NN) {\em
query} centred at $\omega\in\Omega$, where $k\in\N$, can be reduced to a sequence of range queries. This  is discussed in detail in \cite{chavez:01}, Sect. 5.2.

A workload is {\em inner} if $X=\Omega$ and {\em outer} if $\vert X\vert\ll\vert\Omega\vert$. Most workloads of practical interest are outer ones. Cf. \cite{PeSt06}.

\section{Hierarchical tree index structures}

An {\em access method} is an algorithm that correctly answers every range query. Examples of access methods are given by {\em indexing schemes}.
In particular, a {\em hierarchical tree-based indexing scheme} is a sequence of refining partitions of the domain labelled with a finite rooted tree. (For simplicity, we will assume all trees to be binary: this is not really restrictive.) Cf. Figure \ref{fig:tree}.
Such a scheme takes storage space $O(n)$. 

\begin{figure}[ht]
\centering
\scalebox{0.3}[0.3]{
\includegraphics
{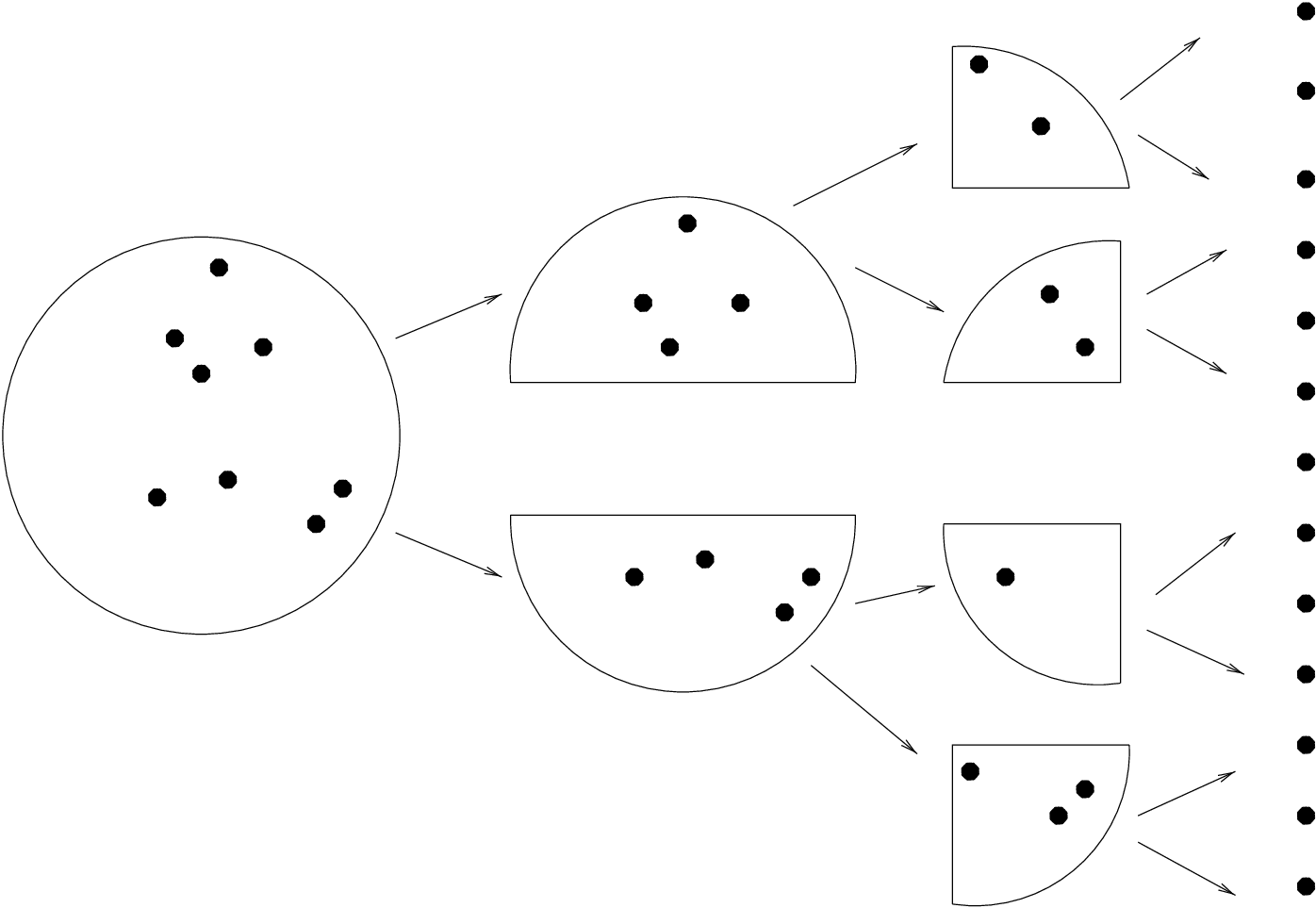}}
\caption{A refining sequence of partitions of $\Omega$.}
\label{fig:tree}
\end{figure}

To process a range query ${\mathcal B}_{\ve}(\omega)$, we traverse the tree recursively to the leaf level. Once a leaf $B$ is reached, its contents (datapoints $x\in X\cap B$) are accessed, and the condition $x\in {\mathcal B}_{\ve}(\omega)$ verified for each one of them. 

Of main interest is what happens at each internal node $C$. Let us identify $C$ with the corresponding element $C\subseteq\Omega$ of the partition, and 
suppose that $A$ and $B$ are child nodes of $C$, 
so that $C=A\cup B$. A branch descending from $B$ can be pruned provided ${\mathcal B}_\ve(\omega)\cap B =\emptyset$, because then datapoints contained in $B$ are of no further interest. This is the case where one can certify that $\omega$ is not contained in the $\ve$-neighbourhood of $B$:
\[\omega\notin B_\ve=\{x\in\Omega\colon \rho(x,B)<\ve\}.\] (Cf. Fig. \ref{fig:split}, l.h.s.) 
Similarly, if $\omega\notin A_\ve$, then the sub-tree descending from $A$ can be pruned. However, if 
the open ball ${\mathcal B}_{\ve}(\omega)$ meets both $A$ and $B$ or, 
equivalently, 
$\omega$ belongs to the intersection of $\ve$-neighbourhoods of $A$ and $B$, pruning is impossible and 
the search branches out. (Cf. Fig. \ref{fig:split}, r.h.s.) 

\begin{figure}[ht]
\centering
\scalebox{0.14}[0.14]{
\includegraphics{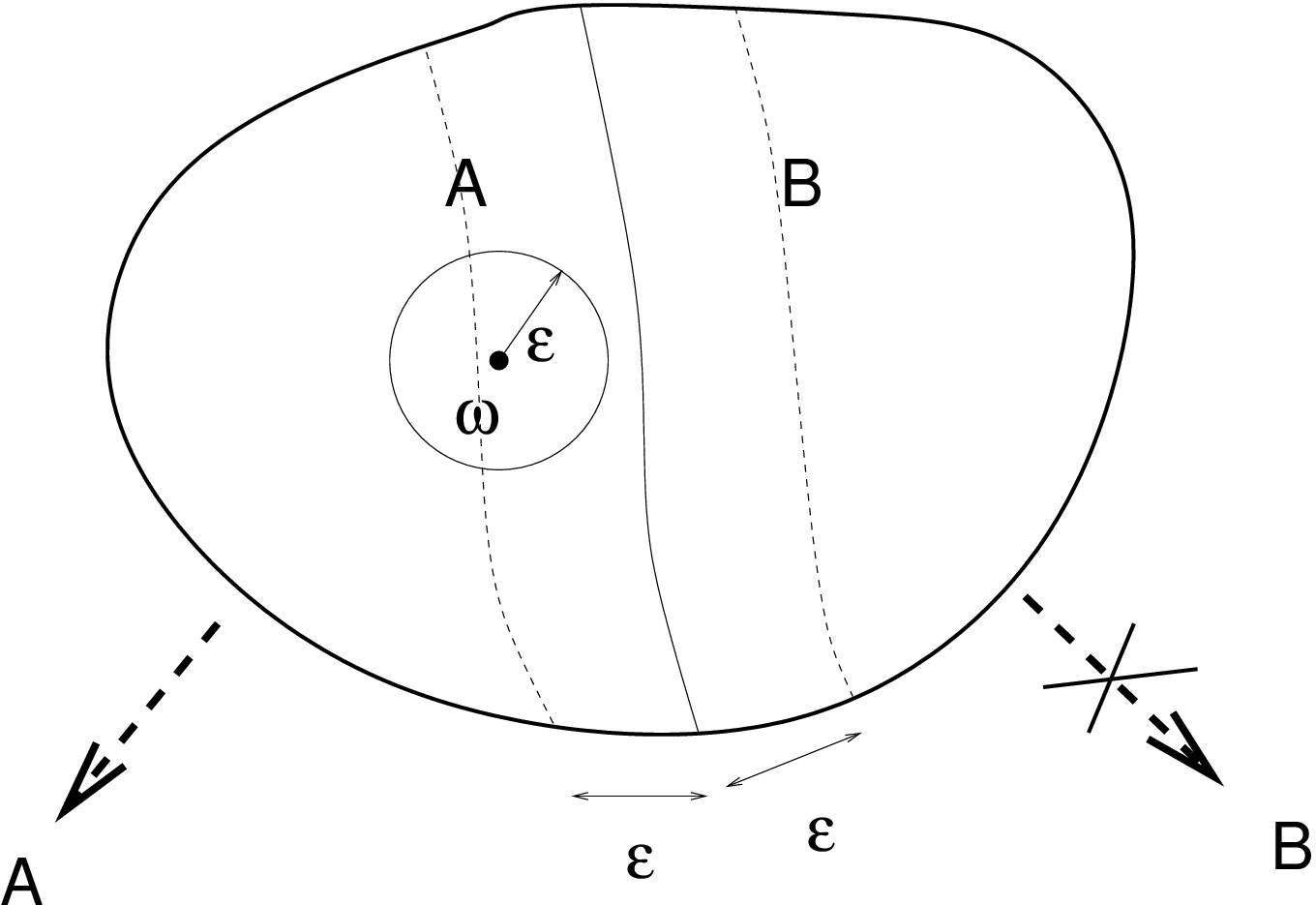}}
\hskip 1.5cm
\scalebox{0.14}[0.14]{
\includegraphics{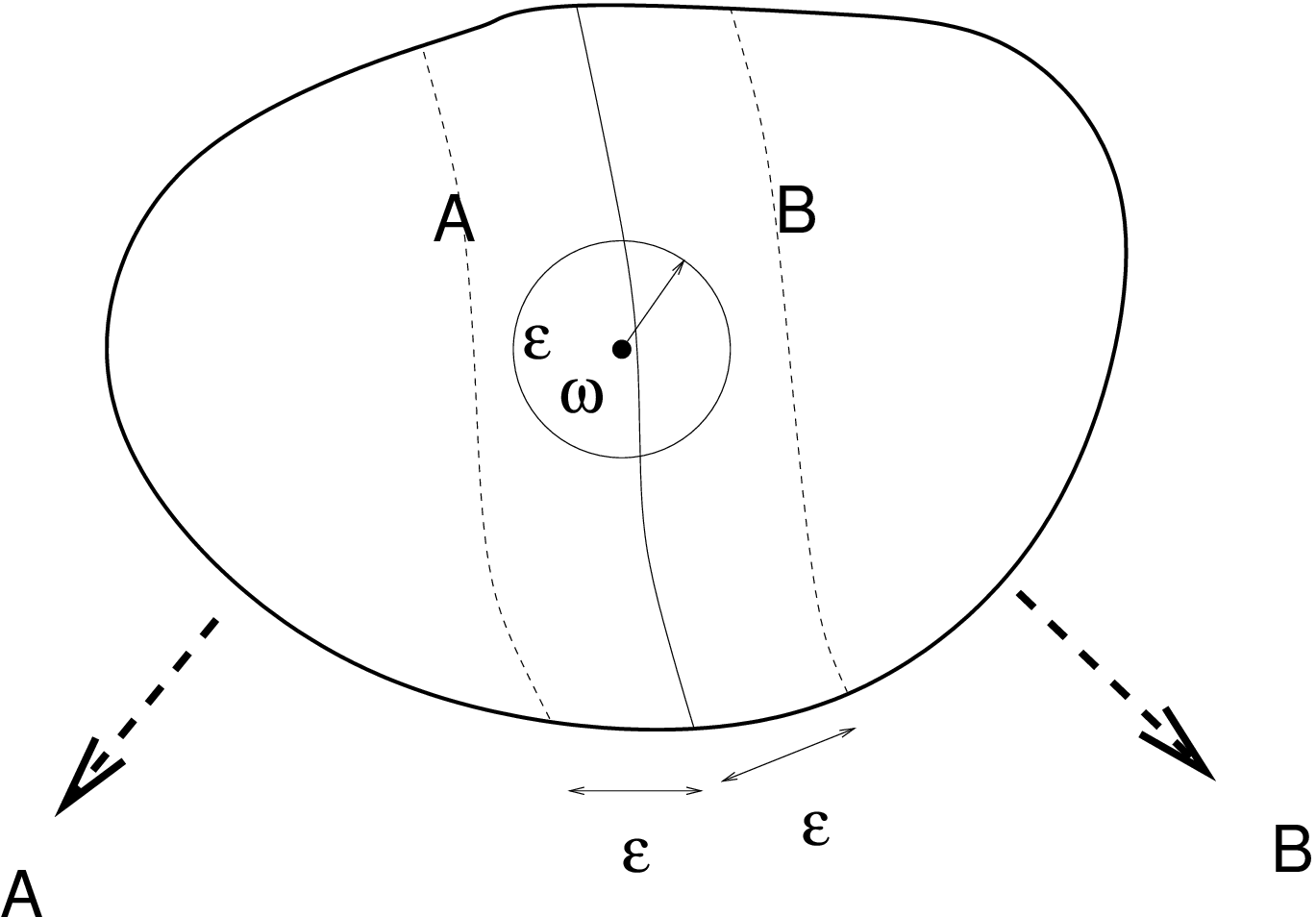}}
\caption{Pruning is possible (l.h.s.), and impossible (r.h.s.).}
\label{fig:split}
\end{figure}

In order to efficiently certify that ${\mathcal B}_\e(\omega)\cap B =\emptyset$, one employs the technique of 
{\em decision functions}. A function $f\colon \Omega\to \R$ is called {\em 1-Lipschitz} if
\[\forall x,y\in\Omega,~~\vert f(x) - f(y)\vert \leq \rho(x,y).\]

Assign to every internal mode $C$ a 1-Lipschitz function $f=f_C$ 
so that $f_C\upharpoonright B \leq 0$ and $f_C\upharpoonright A \geq 0$. 
It is easily seen that $f_C\upharpoonright B_{\ve}<\ve$, and so the fact that 
\fbox{$f_C(\omega)\geq\ve$} serves as a certificate for 
\fbox{${\mathcal B}_\ve(\omega)\cap B =\emptyset$}, assuring that a sub-tree descending from $B$ can be pruned. Similarly, if $f_C(\omega)\leq-\ve$, the sub-tree descending from $A$ can be pruned. 

\begin{figure}[ht]
\centering
\scalebox{0.75}[0.75]{
\includegraphics[width=.5\textwidth]{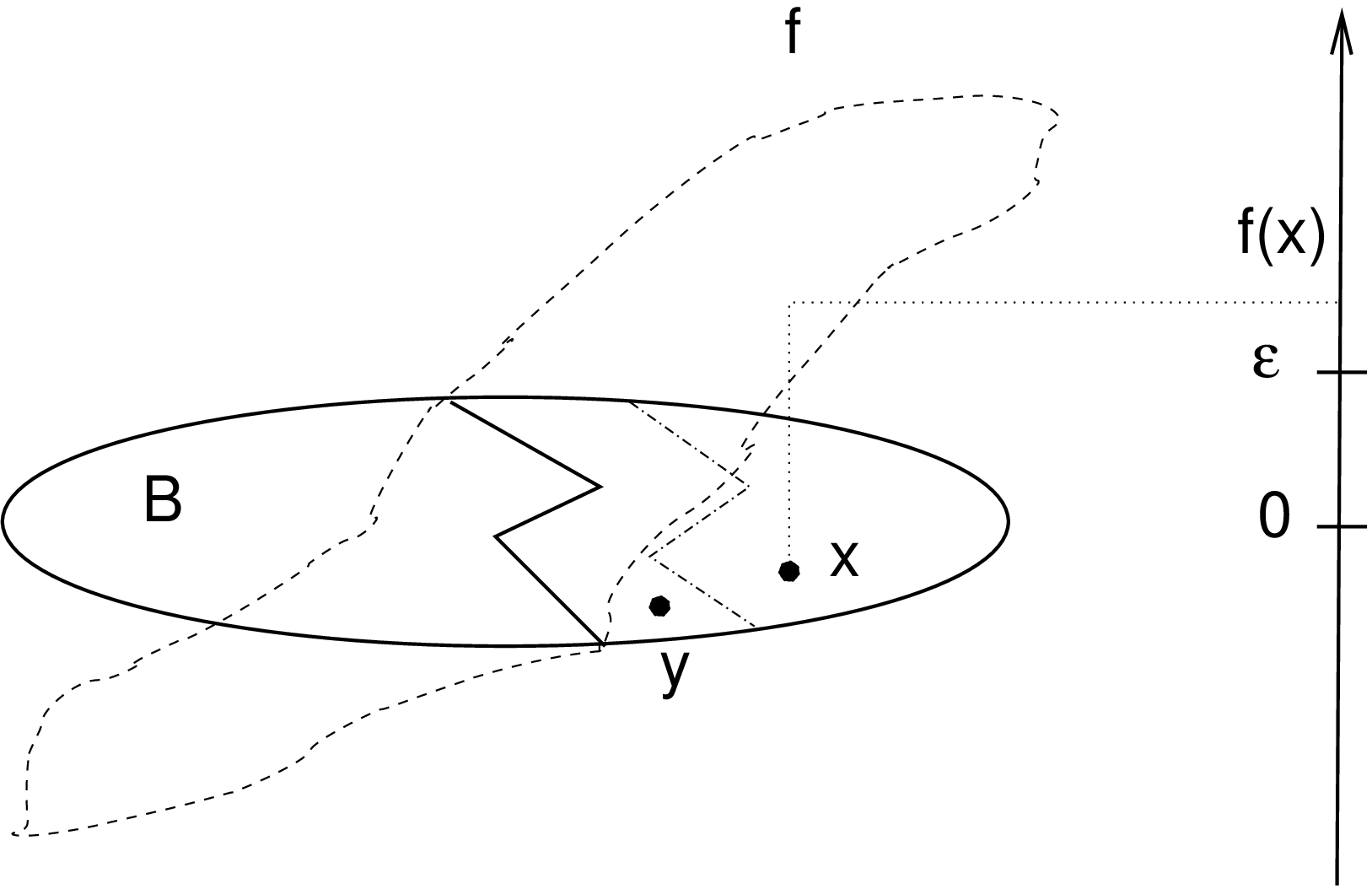}}
\caption{Graph of a decision function $f=f_C$.}
\label{fig:certif}
\end{figure}

Of course, decision functions should have sufficiently low computational complexity in order for the indexing scheme to be efficient. 

A hierarchical indexing structure employing 1-Lipschitz decision functions at every node is known as a {\em metric tree.}

\section{Metric trees}
Here is a formal definition.
A metric tree for a metric similarity workload $(\Omega,\rho,X)$ 
consists of 
\begin{itemize}
\item
a finite binary rooted tree $\mathcal T$,
\item
a collection of (possibly partially defined) real-valued $1$-Lipschitz functions $f_t\colon B_t\to\R$ for every inner node $t$ (decision functions), where $B_t\subseteq\Omega$,
\item a collection of {\em bins} $B_t\subseteq\Omega$ for every leaf node $t$, containing pointers to elements $X\cap B_t$,
\end{itemize}
so that
\begin{itemize}
\item $B_{{\mathrm{root}}({\mathcal T})}=\Omega$,
\item for every internal node $t$ and child nodes $t_-,t_+$, one has
$B_t\subseteq B_{t_-}\cup B_{t_+}$,
\item $f_t\upharpoonright B_{t_-}\leq 0$, $f_t\upharpoonright B_{t_+}\geq 0$.
\end{itemize}

When processing a range query ${\mathcal B}_\e(\omega)$, 
\begin{itemize}
\item $t_{-}$ is accessed $\iff$ $f_t(\omega)<\ve$, and
\item $t_{+}$ is accessed $\iff$ $f_t(\omega)>-\ve$.
\end{itemize}

Here is the search algorithm in pseudocode.

\begin{algorithm} $\,$
\label{alg:search}
\newcommand{\keyw}[1]{{\bf #1}}
\begin{tabbing}
\quad \=\quad \=\quad \=\quad\=\quad\=\quad\=\quad\=\quad\=\quad\kill
\keyw{on input} $(\omega,\ve)$ \keyw{do} \\
\> \> set $A_0=\{{\mathrm{root}}({\mathcal T}) \}$ \\
\>\>\keyw{for} each $i=0,1,\ldots,{\mathrm{depth}}({\mathcal T})-1$ \keyw{do} \\
\>\>\>\>\keyw{if} $A_i\neq\emptyset$ \\
\>\>\>\>\keyw{then} for each $t\in A_i$ \keyw{do} \\
\>\>\>\>\>\>\keyw{if} $t$ is an internal node \\
\>\> \>\>\>\>\keyw{then} \keyw{do}\\
\>\>\>\>\>\>\>\> \keyw{if} $f_t(\omega)<\ve$\\
\>\>\>\>\>\>\>\>\keyw{then} $A_{i+1}\leftarrow A_{i+1}\cup \{ t_{-}\}$ \\
\>\>\>\>\>\>\>\> \keyw{if} $f_t(\omega)>-\ve$\\
\>\>\>\>\>\>\>\>\keyw{then} $A_{i+1}\leftarrow A_{i+1}\cup \{ t_{+}\}$ \\
\>\>\>\>\>\>\keyw{else} \keyw{for} each $x\in B_t$ \keyw{do}\\
\>\>\>\>\>\>\>\>\keyw{if} $x\in {\mathcal B}_{\ve}(\omega)$ \\
\>\>\>\>\>\>\>\>\keyw{then} $A\leftarrow A\cup\{x\}$\\
\keyw{return} $A$
\end{tabbing} \qed
\end{algorithm}

Under our assumptions on the metric tree, it can be proved (cf. \cite{PeSt06}, Theorem 3.3) that
Algorithm \ref{alg:search} correctly answers every range similarity query for the workload $(\Omega,\rho,X)$, and so together with an indexing scheme forms an access method.

\section{Examples of metric tree indexing schemes}

\begin{example}[$vp$-tree]
\label{ex:vptree}
The {\em vp-tree} \cite{Yan} uses decision functions of the form
\[f_t(\omega)=(1/2)(\rho(x_{t_+},\omega)-\rho(x_{t_-},\omega)),\]
where $t_{\pm}$ are two children of $t$ and $x_{t_{\pm}}$ are
the {\em vantage points} for the node $t$.
\end{example}

\begin{example}[$M$-tree]
\label{ex:mtree}
The {\em M-tree} \cite{CPZ97} employs decision
functions 
\[f_t(\omega)=\rho(x_t,\omega)-\sup_{\tau\in B_t}
\rho(x_t,\tau),\]
where 
$B_t$ is a block corresponding to the node $t$, $x_t$ is a datapoint
chosen for each node $t$, and suprema on the r.h.s. are precomputed and stored.
\end{example}

For differing perspectives on metric trees, see \cite{PeSt06,chavez:01}. Each of the books \cite{samet,santini,zezula:06} is an excellent reference to indexing structures in metric spaces.

\section{Curse of dimensionality}
In recent years the research emphasis 
has shifted away from {\em exact} towards {\em approximate} similarity search:
 
\begin{itemize}
\item
given $\e>0$ and $\omega\in\Omega$, return a point $x\in X$ that is [with confidence $>1-\delta$] at a distance $<(1+\e)d_{NN}(\omega)$ from $\omega$.
\end{itemize}

This has led to 
many impressive achievements, particularly \cite{KOR,IM:98}, see also the survey \cite{indyk:04} and Chapter 7 in \cite{vempala}.
At the same time, research in exact similarity search, especially concerning deterministic algorithms, has slowed down. At a theoretical level, the following unproved conjecture helps to keep research efforts in focus.

\begin{conjecture}
[The curse of dimensionality conjecture, 
cf. \cite{indyk:04}]
Let $X\subseteq\{0,1\}^d$ be a data\-set with $n$ points, where the Hamming
cube $\{0,1\}^d$ is equipped with the Hamming ($\ell^1$) distance:
\[d(x,y) = \sharp \{i\colon x_i\neq y_i\}.\]
Suppose $d=n^{o(1)}$, but $d=\omega(\log n)$. (That is, the number of points in $X$ has intermediate growth with regard to the dimension $d$: it is  superpolynomial in $d$, yet subexponential.)
Then any data structure for exact nearest neighbour search in $X$,
with $d^{O(1)}$ query time, must use $n^{\omega(1)}$ space within the {\em cell probe model} of computation.
\end{conjecture}

The best lower bound currently known is $O(d/\log\frac{sd}{n})$, where $s$ is the number of cells used by the data structure \cite{PT}. In particular, this implies the earlier bound $\Omega(d/\log n)$ for polynomial space data structures \cite{barkol:00}, as well as the bound $\Omega(d/\log d)$ for near linear space (namely $n\log^{O(1)}n$). See also \cite{AIP,PTW,PTW2}. A general reference for the cell probe model of computation is  \cite{miltersen}, while in the context of similarity search the model is discussed in \cite{pestov:2011b}.


\section{Concentration of measure}

As in \cite{CPZ}, we assume the existence of an unknown probability measure $\mu$ on $\Omega$, such that both datapoints $X$ and query points $ \omega $ are being sampled with regard to $\mu$. 

On the one hand, this assumption is open to debate: for instance, it is said that in a typical university library most books (75 {\%} or more) are never borrowed a single time, so it is reasonable to assume that the distribution of queries in a large dataset will be skewed equally heavily away from data distribution. On the other hand, there is no obvious alternative way of making an apriori assumption about the query distribution, and in some situations the assumption makes sense indeed, e.g. 
in the context of a large biological database where a newly-discovered protein fragment has to be matched against every previously known sequence. 

The triple $(\Omega,\rho,\mu)$ is known as a {\em metric space with measure}. This concept opens the way to systematically using the {\em phenomenon of concentration of measure on high-dimensional structures}, 
also known as the {\em ``Geometric Law of Large Numbers''} \cite{MS,Le}.
This phenomenon can be informally summarized as follows:

\begin{quote}
{for a typical ``high-dimensional'' structure $\Omega$, if $A$ is a subset
containing at least half of all points, then the 
measure of the $\ve$-neighbourhood
$A_\ve$ of $A$ is overwhelmingly close to $1$ 
already for small $\ve>0$. }
\end{quote}

Here is a rigorous way for dealing with the phenomenon. 
Define the {\em concentration function} $\alpha_{\Omega}$ of a metric space with measure $\Omega$ by

\[\alpha_{\Omega}(\ve)=
\left\{
\begin{array}{ll} \frac 12, & \mbox{if $\ve=0$,} \\
1-\inf\left\{\mu\left(A_\e\right) \colon
A\subseteq\Omega, ~~ \mu(A)\geq\frac 12\right\}, &
\mbox{if $\ve>0$.}
\end{array}\right.
\]

The value of $\alpha_{\Omega}(\ve)$ gives un upper bound on the measure of the complement to the $\ve$-neighbourhood $A_{\ve}$ of every subset $A$ of measure $\geq 1/2$.

For high-dimensional spaces the values of the concentrataion function often admit gaussian upper bounds of the form
\begin{equation}
\label{eq:gaussianbds}
\alpha_{\Omega}(\ve)= \exp(-\Theta(d)\ve^2),
\end{equation}
where $d$ is a dimension parameter. For instance, the concentration function of the $d$-dimensional Hamming cube $\{0,1\}^d$ with the normalized Hamming metric and uniform measure satisfies a Chernoff bound $\alpha(\e)\leq \exp(-2\e^2d)$, cf Fig. \ref{fig:hamming50}.

\begin{figure}[ht]
\centering
\scalebox{0.5}[0.5]{
\includegraphics
{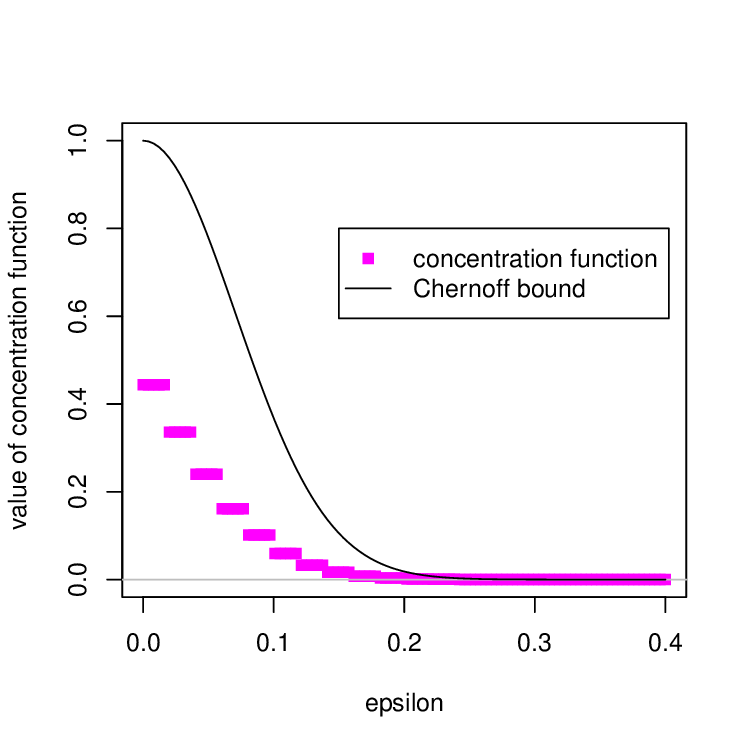}}
\caption{Concentration function of $\{0,1\}^{50}$ {\em vs} Chernoff bound.}
\label{fig:hamming50}
\end{figure}
%

%

Similar bounds hold for Euclidean spheres $\s^n$, cubes $\I^n$, and many other structures of both continuous and discrete mathematics, equipped with suitably normalized distances and canonical probability measures.
The concentration phenomenon can be now expressed by saying that for ``typical'' high-dimensional metric spaces with measure, $\Omega$, the concentration function $\alpha_{\Omega}(\ve)$ drops off sharply as $d\to\infty$ \cite{MS,Le}.

If now $f\colon\Omega\to\R$ is a $1$-Lipschitz function, denote $M=M_f$ the median value of $f$, that is, a (non-uniquely defined) real number with the property that each of the events $[f\geq M]$ and $[f\leq M]$ occurs with probabiity at least half. One can prove without much difficulty:
\begin{equation}
\mu\{x\in\Omega\colon \abs{f(x)-M_f}>\ve\}<2\alpha_{\Omega}(\ve).\end{equation}
Thus, every one-Lipschitz function on a high-dimensional metric space with measure concentrates near one value.

\section{Workload assumptions}
\label{s:model}

Here are our standing assumptions for the rest of the article.
Let $(\Omega,\rho,\mu)$ be a domain equipped with a metric $\rho$ and a probability measure $\mu$. We assume that the expected distance between two points of $\Omega$ is normalized so as to become asymptotically constant:
\begin{equation}
{\mathbb{E}}\,\rho(x,y) = \Theta(1).
\label{eq:exp}
\end{equation}
We further assume that $\Omega$ has ``concentration dimension $d$'' 
in the sense that the concentration function $\alpha_{\Omega}$ is gaussian with exponent $\Theta(d)$;
\begin{equation}
\alpha_{\Omega}(\ve)=\exp\left(-\Theta(\ve^2 d) \right).
\label{eq:dim}
\end{equation}
(This approach to intrinsic dimension is developed in \cite{pestov:08}.)

A dataset $X\subseteq\Omega$ contains $n$ points, where $n$ and $d$ are related as follows:
\begin{eqnarray}
n&=&d^{\omega(1)},\label{eq:superpol} \\
d&=&\omega(\log n). \label{eq:subexp}
\end{eqnarray}
In other words, asymptotically $n$ grows faster than any polynomial function $Cd^k$, $C>0$, $k\in \N$, but slower than any exponential function $e^{cd}$, $c>0$. (An example of such rate of growth is $n=2^{\sqrt d}$.) For the purposes of asymptotic analysis of search algorithms such assumptions are natural \cite{indyk:04}.

Datapoints are modelled by a sequence of i.i.d. random variables distributed according to the measure $\mu$:
\[X_1,X_2,\ldots,X_n\sim\mu.\]
The instances of datapoints will be denoted with corresponding lower case letters $x_1,x_2,\ldots,x_n$.

Finally, the query centres $\omega\in\Omega$ follow the same distribution $\mu$:
\[\omega\sim\mu.\]

%
 

\section{Query radius}
\label{s:branching}

It is known that in high-dimensional domains the distance to the nearest neighbour is approaching the average distance between two points (cf. e.g. \cite{BGRS} for a particular case). This is a consequence of concentration of measure, and the result can be stated and proved in a rather general situation.
Denote $\ve_{NN}(\omega)$ the distance from $\omega\in\Omega$ to the nearest point in $X$. The function $\ve_{NN}$ is easily verified to be
$1$-Lipschitz, and so concentrates near its median value. From here, one deduces:

\begin{lemma}
\label{l:radius}
Under our assumptions on the domain $\Omega$ and a random sample $X$, with confidence approaching $1$ one has for all $\e$
\[\mu\left\{\omega\colon \abs{\ve_{NN}(\omega)-{\mathbb{E}}\,\rho(x,y)}>\e\right\} < \exp(-\Theta(\e^2 d)).\]
\qed
\end{lemma}

\begin{remark}
The result should be understood in the asymptotic sense, as follows. We deal with a family of domains $\Omega_d$, $d\in\N$, and the sampling is performed in each of them in an independent fashion, so that ``confidence'' refers to the probability that the infinite sample path belonging to the infinite product
\[\Omega_1^{n_1}\times\Omega_2^{n_2}\times\ldots\times\Omega_{d}^{n_d}\times\ldots\]
satisfies the desired properties.
\end{remark}

For a proof of Lemma \ref{l:radius}, see Appendix A in \cite{pestov:2011b}.

This effect is already noticeable in medium dimensions.
Let us draw a dataset $X$ with $10,000$ points randomly from the Euclidean cube $[0,1]^{50}$ with regard to the uniform measure. Then, with respect to the usual Euclidean distance, 
the median value of the distance to the nearest neighbour is $\e_M =  1.9701$, while the expected value of a distance between two points of $X$, ${\mathbb{E}}d(x,y) =  2.872$. Cf. Fig. \ref{fig:nn_distances50} for the distribution of values of $\ve_{NN}$. 

\begin{figure}[ht]
 \centering
 \includegraphics[width=.4\textwidth]{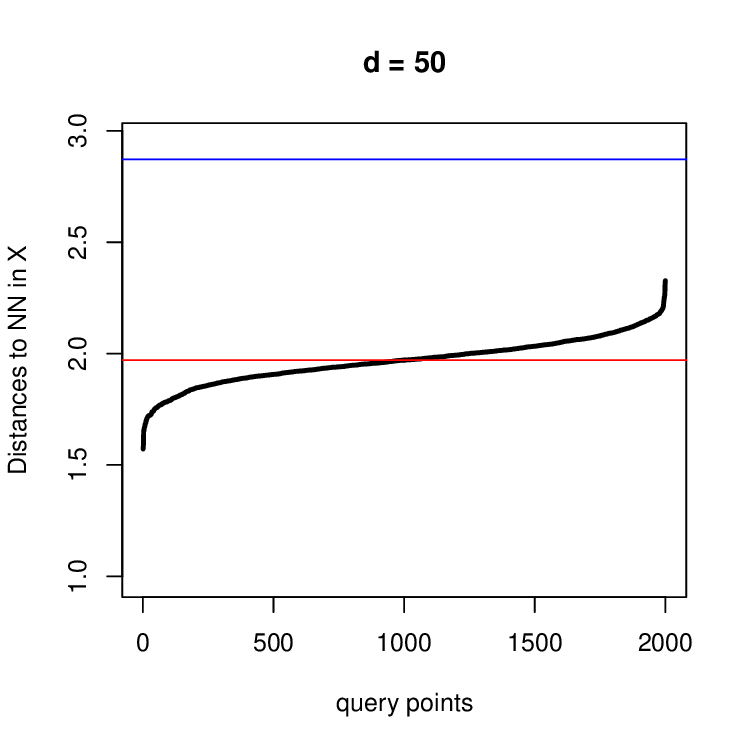}
\caption{Distances from $2,000$ random query points to their nearest neighbours in a dataset of $10,000$ random points in the Euclidean cube $[0,1]^{50}$. The lower horizontal line marks $\e_M =  1.9701$, the upper ${\mathbb{E}}d(x,y) =  2.872$.}
 \label{fig:nn_distances50}
 \end{figure}


\section{A ``naive'' $O(n)$ lower bound}

As a first approximation to our analysis, we present a heuristic argument, allowing {\em linear} in $n$ asymptotic lower bounds on the search performance of a metric tree. 

What happens at an internal node $C$ when a metric tree is being traversed? Note that $C$ itself becomes a metric space with measure if equipped with the metric induced from $\Omega$ and a probability measure $\mu_C$ which is the normalized restriction of the measure $\mu$ from $\Omega$:
\[\mbox{for }A\subseteq C,~~\mu_C(A)=\frac{\mu(A)}{\mu(C)}.\]
Let $\alpha_C$ denote the concentration function of $C$.
Suppose for the moment that our tree is perfectly balanced:
$\mu_C(A)=\mu_C(B)=\frac 12$.
Then the size of the $\ve$-neighbourhood of $A$ is at least $1-\alpha_C(\ve)$, and the same is true of $B_{\ve}$. For all query points $\omega\in C$ except a set of measure
$\leq 2\alpha_C(\ve)$,
the search algorithm \ref{alg:search} branches out at the node $C$. (Cf. Fig. \ref{fig:split3}.)

\begin{figure}[ht]
\centering
\scalebox{0.25}[0.25]{
\includegraphics
{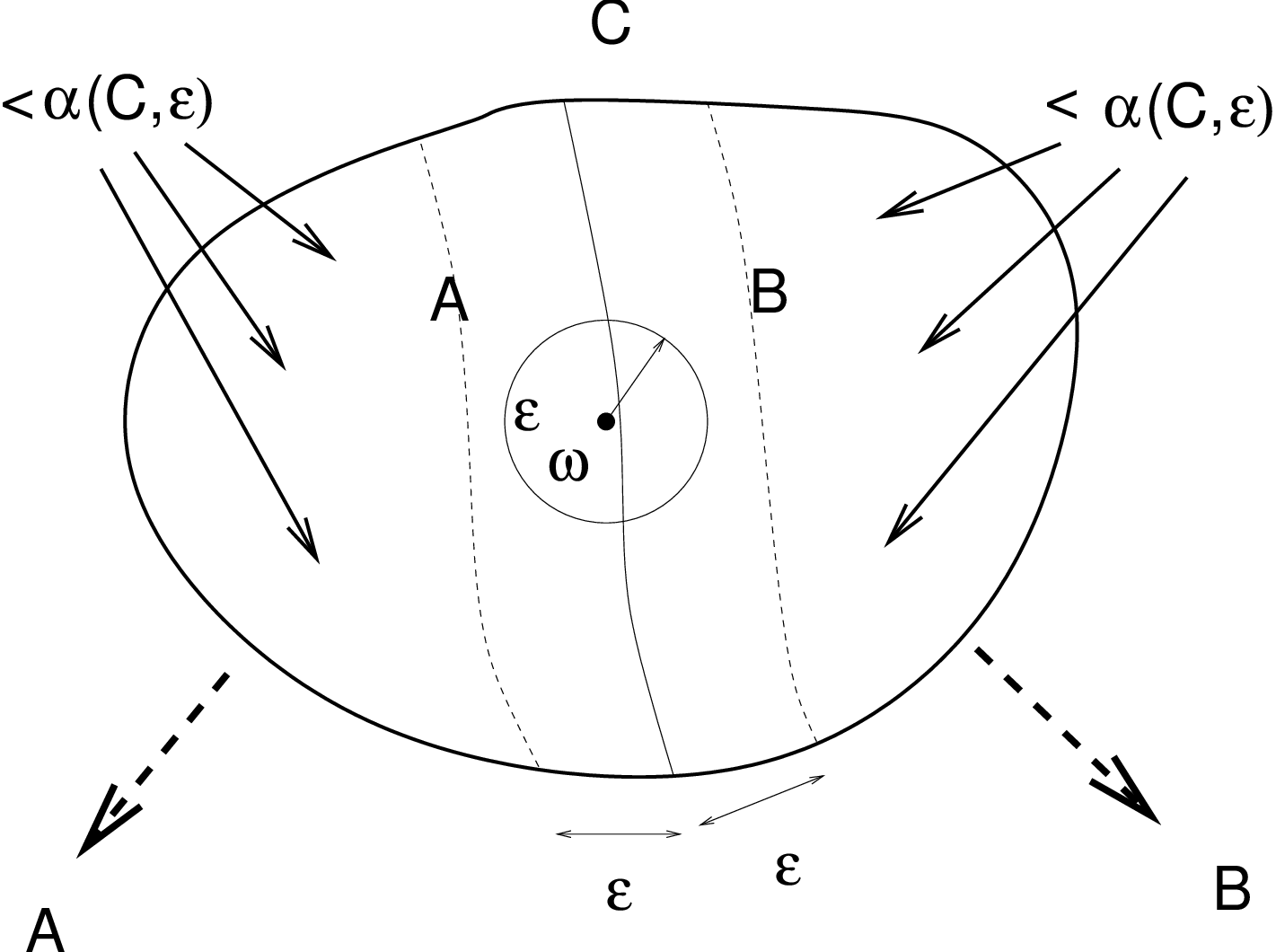}}
\caption{Search algorithm branches out for most query points $\omega$ at a node $C$ if the value $\alpha_C(\ve)$ is small.}
\label{fig:split3}
\end{figure}

\begin{lemma}
Let $C$ be a subset of a metric space with measure $(\Omega,\rho,\mu)$. 
Denote $\alpha_C$ the concentration function of $C$ with regard to the induced metric $\rho\upharpoonright C$ and the induced probability measure $\mu/\mu(C)$. 
Then for all $\ve>0$
\[\alpha_C(\ve) \leq \frac{\alpha_{\Omega}(\ve/2)}{\mu(C)}.\]
\label{l:subspace}
\end{lemma}

\proof
Let $\ve>0$ be any, and let $\delta<\alpha_C(\ve)$. Then there are subsets $D,E\subseteq C$ at a distance $\geq\ve$ from each other, satisfying $\mu(D)\geq\mu(C)/2$ and $\mu(E)\geq\delta\mu(C)$, in particular the measure of either set is at least $\delta\mu(C)$. Since the $\ve/2$-neighbourhoods of $D$ and $E$ in $\Omega$ cannot meet by the triangle inequality, the complement, $F$, to at least one of them, taken in $\Omega$,
has the property $\mu(F)\geq 1/2$, while
$\mu(F_{\ve/2})\leq 1-\delta\mu(C)$, because $F_{\ve/2}$ does not meet one of the two original sets, $D$ or $E$. We conclude: 
$\alpha_{\Omega}(\ve/2)\geq \delta\mu(C)$,
and taking suprema over all $\delta<\alpha_C(\ve)$,
\[\alpha_{\Omega}(\ve/2)\geq \alpha_C(\ve)\mu(C),\]
that is, $\alpha_C(\ve)\leq\alpha_{\Omega}(\ve/2)/\mu(C)$, as required.
\qed

Since the size of the indexing scheme is $O(n)$, a typical size of a set $C$ will be on the order $\Omega\left(n^{-1}\right)$, while $\alpha_{\Omega}(\ve)$ will go to zero as $o\left(n^{-1}\right)$.

Let a workload $(\Omega,\rho,X)$ 
be indexed with a balanced metric tree of depth $O(\log n)$, having $O(n)$ bins of roughly equal $\mu$-measure.
For at least half of all query points, the distance $\ve_{NN}$ to the nearest neighbour in $X$ is at least as large as $\ve_M$, the median NN distance. 
Let $\omega$ be such a query centre.
For every element $C$ of level $t$ partition of $\Omega$, one has, using Lemmas \ref{l:subspace} and \ref{l:radius} and the assumption in Eq. (\ref{eq:dim}),
\[\alpha_C(\ve_M)\leq \frac{\alpha_{\Omega}(\ve_M/2)}{\mu(C)^{-1}}=\Theta(2^t) e^{-\Theta(1)\ve_M^2d}=e^{-\Theta(d)},\]
where the constants {\em do not depend} on a particular internal node $C$.
An argument in Section \ref{s:branching} implies that branching {\em at every internal node} occurs for all $\omega$ except a set of measure 
\[\leq\sharp(\mbox{nodes})\times 2\sup_C\alpha_C(\ve)=O(n^2)e^{-\Theta(d)}=o(1),\]
because $d=\omega(\log n)$ and so $e^{\Theta(d)}$ is superpolynomial in $n$.
Thus, the expected average performance of an indexing scheme as above is linear in $n$. 

There are two problems with this argument. Firstly, it has been observed and confirmed experimentally that unbalanced metric trees can be more efficient than the balanced ones \cite{ChN:05,navarro:2008}. Secondly and more importantly, we have replaced the value of the {\em empirical measure},
\[\mu_n(C)=\frac {\abs C}{n} ,\]
with the value of the underlying measure $\mu(C)$, implicitly assuming that the two are close to each other:
\[\mu_n(C) \approx \mu(C).\]
But the scheme is being chosen {\em after} seeing an instance $X$, and it is reasonable to assume that indexing partitions will take advantage of random clusters always present in i.i.d. data.  
(Fig. \ref{fig:randsq1000} illustrates this point in dimension $d=2$.) Some elements of indexing partitions, while having large $\mu$-measure, may contains few datapoints, and vice versa.

\begin{figure}[ht]
\centering
\includegraphics[width=.3\textwidth,height=.3\textwidth]{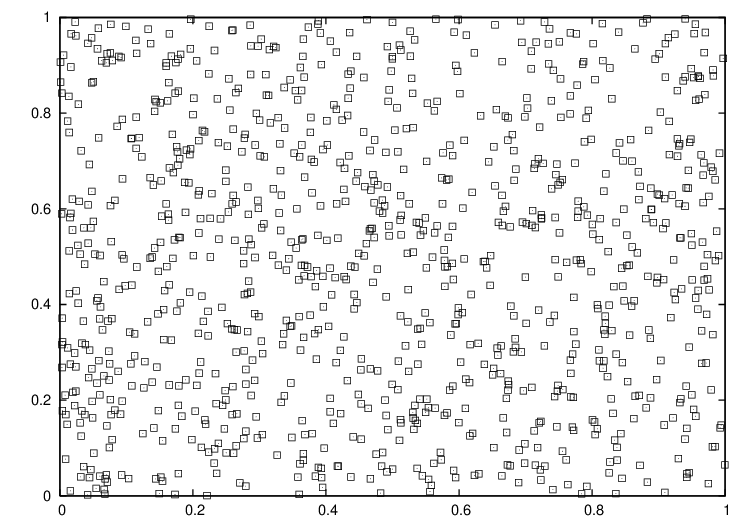}
\caption{$1000$ points randomly and uniformly distributed in the square $[0,1]^2$.}
\label{fig:randsq1000}
\end{figure}

An equivalent consideration is that we only know the concentration function of the domain $\Omega$, but not of a randomly chosen dataset $X$. It seems the  problem of estimating the concentration function of a random sample has not been systematically treated. 

In order to be able to estimate the empirical measure in terms of the underlying distribution, one needs to invoke an approach of statistical learning. 

\section{Vapnik--Chervonenkis theory}

Let $\mathscr C$ be a family of subsets of a set $\Omega$ (a {\em concept class}). One says that a subset $A\subseteq\Omega$ is {\em shattered} by $\mathscr C$ 
if for each $B\subseteq A$ there is $C\in{\mathscr C}$ such that
\[C\cap A = B.\]
%

The {\em Vapnik--Chervonenkis dimension} $\VC({\mathscr C})$ of a class $\mathscr C$ is the supremum of sizes of finite subsets $A\subseteq\Omega$ shattered by $\mathscr C$.

Here are some examples.

\begin{enumerate}
\item
The VC dimension of the class of all Euclidean balls in $\R^d$ is $d+1$. 
\item
The class of all parallelepipeds in $\R^d$ has VC dimension $2d+2$. 
\item
The VC dimension of the class of all balls in the Hamming cube $\{0,1\}^d$ is bounded from above by $d+\lfloor \log_2 d\rfloor$.
\par
(As every ball is determined by its centre and radius, the total number of pairwise different balls in $\{0,1\}^d$ is $d2^d$. Now one uses an obvious observation: the VC dimension of a finite concept class $\mathscr A$ is bounded above by $\log_2\abs{\mathscr A}$.)
\end{enumerate}

Here is a deeper result.

\begin{theorem}[Goldberg and Jerrum \cite{GJ}, Theorem 2.3]  
Let 
\[{\mathscr F}=\{x\mapsto f(\theta,x)\colon\theta\in \R^s\}\]
be a parametrized class of $\{0,1\}$-valued functions. Suppose that, for each input $x\in\R^n$, there is an algorithm that computes $f(\theta,x)$, and this computation takes no more than $t$ operations of the following types:
\begin{itemize}
\item the arithmetic operations $+,-,\times$ and $/$ on real numbers,
\item jumps conditioned on $>$, $\geq$, $<$, $\leq$, $=$, and $\neq$ comparisons of real numbers, and
\item output $0$ or $1$.
\end{itemize}
Then $\VC({\mathscr F})\leq 4s(t+2)$. \qed
\label{th:gj}
\end{theorem}

Here is a typical result of statistical learning theory, which we quote from \cite{vidyasagar:2003}, Theorem 7.8.

\begin{theorem}
\label{th:learning}
Let $\mathscr C\subseteq 2^{\Omega}$ be a concept class of finite VC dimension, $d$. 
Then for all $\epsilon,\delta>0$ and every probability measure $\mu$ on $\Omega$,
if $n$ datapoints in $X$ are drawn randomly and independently acoording to $\mu$, then with confidence $1-\delta$
\[\forall C\in {\mathscr C},~~\left\vert \mu(C)-\frac{\abs{X\cap C}}{n} \right\vert <\epsilon,\]
provided 
\[n\geq  \max\left\{\frac{8d}{\e}\lg\frac{8e}{\e},\frac 4{\e}\lg\frac{2}{\delta} \right\}.
\]
\end{theorem}


Let $\mathscr F$ be a class of (possibly partially defined) real-valued functions on $\Omega$. Define ${\mathscr F}_{\geq}$ as the family of all sets of the form
\[\{\omega\in{\mathrm{dom}\, f}\colon f(\omega)\geq a\},~~a\in\R.\]
The value of $\VC({\mathscr F}_{\geq})$
is bounded above by the Pollard dimension (pseudodimension) of $\mathscr F$ (cf. \cite{vidyasagar:2003}, 4.1.2), but is in general smaller. 

\begin{example}[Pivots]
\label{ex:pivots}
If $\mathscr F$ is the class of all distance functions to points of $\R^d$, then $\VC({\mathscr F}_{\geq})=d+1$. (The family ${\mathscr F}_{\geq}$ consists of complements to open balls, and the VC dimension is invariant under proceeding to the complements.) For the Hamming cube, $\VC({\mathscr F}_{\geq})\leq d+\lfloor \log_2 d\rfloor$.
\end{example}

%



\begin{example}[$vp$-tree] See Example \ref{ex:vptree}.
%
If $\Omega=\R^d$, then ${\mathscr F}_{\geq}$ consists of all half-spaces, and the VC dimension of this family is well known to equal $d+1$.
\end{example}

\begin{example}[$M$-tree] See Example \ref{ex:mtree}. 
The dimension estimates are the same as in Example \ref{ex:pivots}.
\end{example}

For both schemes, 
if $\Omega=\R^d$ or $\{0,1\}^n$, then $\VC(({\mathscr F})_{\geq})$
equals $d+1$.
A similar conclusion holds for the Hamming cube.

%
%
%
%
%

\section{Rigorous lower bounds}

In this Section we prove the following theorem under general assumptions of Section \ref{s:model}.

\begin{theorem}
\label{th:main}
Let the domain $\Omega$ equipped with a metric $\rho$ and probability measure $\mu$ have concentration dimension $\Theta(d)$ (cf. Eq. (\ref{eq:dim})) and expected distance between two points ${\mathbb{E}}d(x,y)=1$.
Let $\mathscr F$ be a class of all 1-Lipschitz functions on the domain $\Omega$ that can be used as decision functions for metric tree indexing schemes of a given type. Suppose $\VC({\mathscr F}_{\geq})=o(n^{1/4}/\log^2n)$.
Let $X=\{x_1,x_2,\ldots,x_n)$ be an instance of an i.i.d. random sample of $\Omega$ following the distribution $\mu$, where $d=n^{o(1)}$ and $d=\omega(\log n)$. 
Then an optimal metric tree indexing scheme for the similarity workload $(\Omega,\rho,X)$ has expected average runtime $\Omega(n^{1/4})$.
\end{theorem}

The following is a direct application of Lemma 4.2 in \cite{pestov:00}.

\begin{lemma}[``Bin Access Lemma'']
\label{l:bal}
Let $\varepsilon>0$ and $m\geq 4$ be such that $\alpha_{\Omega}(\ve)\leq m^{-1}$, and let $\gamma$ be a collection of 
subsets $A\subseteq\Omega$ of measure 
$\mu(A)\leq m^{-1}$ each, satisfying $\mu(\cup\gamma)\geq 1/2$. 
Then the $2\varepsilon$-neighbourhood of every point 
$\omega\in \Omega$, apart from a set of measure at most
$\frac 12 m^{-\frac 12}$,
meets at least $\frac 1 2 m^{\frac 12}$
elements of $\gamma$.
\end{lemma}

Here is the next step in the proof.

\begin{lemma}
Let $\mathscr F$ be a family of real-valued functions satisfying $\VC({\mathscr F}_{\geq})\leq p$. Denote $\mathscr B$ the class of all subsets $B\subseteq\Omega$ appearing as intersections of $\leq h$ sets of the form $[f \gtreqqless a]$, $f\in {\mathscr F}$. Then
\[\VC({\mathscr B})\leq 4hp\log(2hp).\]
\label{l:bins}
\end{lemma}

\proof
Use Th. 4.5 in \cite{vidyasagar:2003}: if $\mathscr A$ is a concept class of VC dimension $\leq p$, then the VC dimension of the class of all sets obtained as intersections of $\leq h$ sets from $\mathscr F$ is bounded by $2h p\log(hp)$.\qed

\proof
We can suppose that 
the expected average depth of a tree traversed is $o(n^{1/4})$, for otherwise there is nothing to prove. 

Using Eq. (\ref{eq:exp}) and Lemma \ref{l:radius}, pick any $\ve^\prime>0$ such that, for sufficiently high values of $d$, for most points $\omega$ (that is, for a set of $\mu$-measure $1-o(1)$) the value of $\ve_{NN}(\omega)$ exceeds $\ve^\prime$. 
Similarly, we can assume that 
query points of $\mu$-measure $1-o(1)$ have the property that their $\ve^\prime$-neighbourhood only meets bins with fewer than $n^{1/4}$ datapoints. 
(Otherwise, already scanning the contents of large bins would result in an expected running time $\Omega(n^{1/4}$.)

Combining the two assumptions together, we deduce that for a set $\Omega^\prime$ of query centres $\omega$ of $\mu$-measure $1-o(1)$ the following are true: (1) the $\ve^\prime$-ball around $\omega$ only meets bins with fewer than $n^{1/4}$ points, and (2) the depth of every search tree beginning with $\omega$ does not exceed $n^{1/4}$.

Let $b=\{t_0,t_1,\ldots,t_k=t\}$ be a branch of the search tree corresponding to a query point $\omega\in\Omega^{\prime}$.
Let $\Omega_b$ denote the set of all $\omega\in\Omega^\prime$ for which the branch $b$ has to be followed. Then $\Omega_b\subseteq B_t$, and so $\Omega_b$ contains fewer than $n^{1/4}$ datapoints. Also, $\Omega_b$ is the intersection of a family of $\leq n^{1/4}$ sets of the form $[f \gtreqqless a]$, $f\in {\mathscr F}$. By Lemma \ref{l:bins} and our assumption on $\mathscr F$, the VC dimension of the collection, $\mathscr B$, of all possible sets $\Omega_b$ emerging in this fashion is $o(n^{1/2}/\log n)$.

Apply Theorem \ref{th:learning} to the concept class $\mathscr B$
with $\ve = n^{-1/2}$. If $n$ is sufficiently large, then with high confidence the $\mu$-measure of every element of $\mathscr B$ does not differ from the empirical measure (which is $\leq n^{-3/4}$) by more than $\ve = n^{-1/2}$. One concludes: with high confidence, the sets $\Omega_q$, $q\in\Omega^\prime$ have $\mu$-measure $\leq 2n^{-1/2}$.

The Bin Access Lemma \ref{l:bal}, applied with $m=2n^{1/2}$ and $\ve=\ve^\prime/2$, implies that for all $\omega\in\Omega^\prime$ the $\ve^\prime$-neighbourhood of $\omega$ meets at least $O(n^{1/4})$ pairwise different sets of the form $\Omega_b$ as above. Since $\mu(\Omega^\prime)=1-o(1)$, this implies the need to traverse on average $\Omega(n^{1/4})$ distinct branches of the search tree, establishing the claim.
\qed

\smallskip
Combining our Theorem \ref{th:main} with Theorem \ref{th:gj} of Goldberg and Jerrum shows that for all practical purposes the expected average performance of metric trees is superpolynomial in dimension of the domain.

\begin{corollary}
Let the domain $\Omega=\R^d$ be equipped with a probability measure $\mu_d$ in such a way that the concentration function of $(\R^d,\mu_d)$ admits a gaussian upper bound and the $\mu_d$-expected value of the Euclidean distance is $\Theta(1)$. Let ${\mathscr F}_d$ denote a class of functions $f(\theta,x)$ on $\R^d$ parametrized with $\theta$ taking values in a space $\R^{{{\mathrm{poly}\,}}(d)}$ and such that computing each value $f(\theta,x)$ takes $d^{O(1)}$ operations of the type described in Thm. \ref{th:gj}.
Let $X$ be an i.i.d. random sample of $\R^d$ according to $\mu_d$, having $n$ points, where $d=n^{o(1)}$ and $d=\omega(\log n)$. 
Then, with confidence asymptotically approaching $1$, an optimal metric tree indexing scheme for the similarity workload $(\Omega,\rho,X)$ whose decision functions belong to the parametrized class $\mathscr F$ has expected average runtime $d^{\omega(1)}$. \qed
\end{corollary}

Three remarks are in order to explain the strength of the above results.

(1) Measures $\mu_d$ satisfying the above assumption include, for instance, the  gaussian distribution, the uniform measure on the unit ball, on the unit sphere, on the unit cube, etc.

(2) A polynomial upper bound on the size of the parameter $\theta$ for $\mathscr F$ is dictated by the obvious restriction that reading off a parameter of superpolynomial length leads to a superpolynomial lower bound on the length of computation. 

(3) In the situations of interest, one can verify that the expected number of datapoints $x\in X$ contained in the smallest query ball meeting $X$ is $O(1)$. For continuous measures on $\R^n$ such as the gaussian measure or the uniform measure on the cube etc., this will be obviously $1$. For the Hamming cube, the upper limit of this number as $d\to\infty$ is bounded by $e\approx 2.7182\ldots$. Thus, the lower bound does not come from the fact that there are simply too many valid near neighbours.

(4) We do not know the answer to the following.

{\bf Question.} Cost of computing the values of decision functions aside, can a dataset $X\subset\{0,1\}^d$, $n=\abs{X}$, $d=\omega(\log n)$, $d=n^{o(1)}$, be indexed with a metric tree performing in time $\mbox{poly}(d)$?

\section{Conclusion}

In this Section, written in response to referee's comments, the author will try to outline his understanding of applicability of the method of proof to other indexing paradigms.

The approach to obtaining lower bounds on performance of indexing schemes adopted in this paper consists in combining simple concentration of measure considerations with the basic techniques of statistical learning (VC theory).
The argument is applicable to the situation of the following kind. 
Let $W=(\Omega,\rho,X)$ denote a similarity workload. An indexing scheme for $W$ consists of a family of real-valued $1$-Lipschitz functions $f_i$, $i\in I$ on $\Omega$, which are in general partially defined: ${\mathrm{dom}}\,(f_i)\subseteq\Omega$.
Given a query $(\omega,\ve)$, where $\omega\in\Omega$ and $\e>0$, the algorithm chooses recursively a sequence of indices $i_n$, based on the previous values $f_{i_k}(\omega)$, $k < n$. 
At some point, the computation is terminated, and the values $f_{i_k}(\omega)$ point at a collection of bins, whose contents are read off. 
The role of the functions $f_i$ is to discard those datapoints (or the entire bins) which cannot possibly answer the query. Namely, if $\abs{f_i(\omega)-f_i(x)}\geq\ve$, then, since $f_i$ is a $1$-Lipschitz function, one has $d(\omega,x)\geq \ve$, and so the point $x$ is irrelevant. All the points (or entire bins) which cannot be discarded are returned and their contents checked against the condition $d(x,\omega)<\ve$. 

On the spaces of high dimension, every $1$-Lipschitz function concentrates sharply near its mean (or median) value. If in addition we assume that the class $\mathscr F$ of all functions used for a particular indexing scheme has a low complexity in the sense of VC dimension, we can conclude that the number of points discarded by every function $f_i$ drops off fast as dimension $d$ of the domain grows, resulting in degrading performance. 

So far, we are aware of essentially two different types of such indexing schemes: metric trees (treated in the present paper) and pivot tables \cite{bustos:03}. For pivots, the methods of the present paper have been subsequently used to derive an expected average performance lower bound $\Omega\left({n}/{d\log n}\right)$ \cite{VolPest09}. It is not clear to the author how to state a more general result from which both estimates would follow, nor whether such a result would be useful in view of lack of other examples.

Even if the cell-probe model has some formal similarities with the metric tree scheme (a hierarchical tree structure, a collection of cells as an indexing scheme, computations performed at each node with a limited number of cells accessed, etc.), it is not clear whether the partially defined functions determined by the  algorithm at each node will be $1$-Lipschitz (they are taking values in the Hamming cube). The examples of implemented indexing schemes for {\em exact} nearest neighbour search known to this author seem to be using $1$-Lipschitz functions, but of course this does not preclude the existence of schemes based on other ideas.

Furthermore, assuming that an indexing scheme consists of a family of $1$-Lipschitz functions whose values are recursively computed by the algorithm  does not necessarily imply that the role of the functions is reduced to certifying that a certain point is not in the $\ve$-ball around the query point. As an example, consider the indexing scheme \cite{clarkson:94} based on a walk on the Delaunay graph of $X$ in $\Omega$ and called {\em spatial approximation} in \cite{navarro:2002}. For every datapoint $x\in X$, the scheme stores a list of datapoints whose Voronoi cells are adjacent to the cell containing $x$. At the search phase, a sequence of datapoints $x_1,x_2,\ldots,x_n$ is chosen, where each $x_{i+1}$ is the closest point to $\omega$ on the list of points Delaunay-adjacent to $x_i$. If choosing $x_{i+1}$ so as to get closer to $\omega$ is impossible, one backtracks. In practice, the scheme performs on par with the state of the art pivot or metric tree based schemes \cite{NR:2009}. We do not know whether our methods can be employed to prove the curse of dimensionality for this particular scheme in the same general setting.

It appears that attempting to extend the method to randomized, approximated NN search stands no chance either. Firstly, the dimensionality reduction-type methods often present in randomized algorithms for approximate search \cite{KOR,IM:98,AIP} mean that instead of $1$-Lipschitz functions, one is using what may be called ``probably approximately $1$-Lipschitz'' ones. For instance, a random projection from a high-dimensional Euclidean space to a subspace of smaller dimension, appropriately rescaled, will have the property that for most pairs of points $x,y$ the distance between them is approximately preserved, to within a factor of $1\pm\e$. This property in itself is a consequence of concentration of measure, but
such maps do not exhibit a strong concentration property, rendering our methods inapplicable. 

Chapter 4 in \cite{zezula:06} discusses algorithms for approximate similarity search based on a traditional metric tree, equipped with $1$-Lipschitz decision functions, but employing agressive pruning, either randomized or deterministic. 
Even here, our proof does not seem to be readily transferable. Indeed, it is based on the basic premise that {\em every bin meeting the $\ve$-neighbourhood of the query point needs to be examined in a deterministic fashion}. A randomized algorithm, on the contrary, avoids opening bins which are deemed unlikely to contain relevant datapoints. Experiments confirm that some of the algorithms in question perform up to 300 times faster than the corresponding algorithms for exact search using the same indexing structure ({\em loc.cit.}), and provide a circumstantial evidence that the situation here is indeed fundamentally different and possibly not amenable to the same methods of analysis.



While the setting of artificially high-dimensional synthetic i.i.d. data fed to a scheme is not realistic, our results provide a theoretical validation to the known simulation results on the poor performance in medium to high dimensions of metric-tree type indexing schemes, such as SS tree \cite{WJ} and SR tree \cite{KS}, on such data inputs.

Some data practitioners believe that the intrinsic dimension of real-life datasets does not exceed as few as perhaps seven or ten dimensions. A deeper understanding of underlying geometry of workloads and its interplay with compleixty is called for in order to learn to detect and use this low dimensionality efficiently, and asymptotic analysis of algorithm performance in an artificial setting of very high dimensions is contributing towards this goal.

\end{document}